\documentclass[11pt,twocolumn]{article}

\usepackage{amssymb, amsmath,algorithm,algorithmic,amsthm}
\usepackage{graphicx,epsfig,subfigure}
\usepackage[left=1in,top=1in,right=1in,bottom=1in,nohead]{geometry}

\newcommand{\weakaddition}{{\tt Weak-Any}}
\newcommand{\strongaddition}{{\tt Strong-Any}}
\newcommand{\weakgreedy}{{\tt Weak-Greedy}}
\newcommand{\stronggreedy}{{\tt Strong-Greedy}}

\newcommand{\an}{$(k,1)$}
\newcommand{\anonvar}{$(k,\ell)$}
\newcommand{\secondvar}{$\ell$}
\newtheorem{problem}{{Problem}}
\newtheorem{theorem}{{Theorem}}
\newtheorem{example}{{Example}}
\newtheorem{definition}{{Definition}}
\newtheorem{lemma}{{Lemma}}

\newtheorem{proposition}{{Proposition}}

\newcommand{\squishlist}{
   \begin{list}{$\bullet$}
    { \setlength{\itemsep}{0pt}      \setlength{\parsep}{3pt}
      \setlength{\topsep}{3pt}       \setlength{\partopsep}{0pt}
      \setlength{\leftmargin}{1.5em} \setlength{\labelwidth}{1em}
      \setlength{\labelsep}{0.5em} } }

\newcommand{\squishlisttwo}{
   \begin{list}{$\bullet$}
    { \setlength{\itemsep}{0pt}    \setlength{\parsep}{0pt}
      \setlength{\topsep}{0pt}     \setlength{\partopsep}{0pt}
      \setlength{\leftmargin}{2em} \setlength{\labelwidth}{1.5em}
      \setlength{\labelsep}{0.5em} } }

\newcommand{\squishend}{
    \end{list}  }

\begin{document}

\title{Anonymizing Graphs}
\author{
Tom\'{a}s Feder\\Stanford University \and Shubha U. Nabar\\Stanford University \and Evimaria Terzi\\IBM Almaden
 }
\date{}
 
\maketitle

\begin{abstract}
Motivated by recently discovered privacy attacks on social networks, we study the problem of anonymizing the underlying graph of interactions in a social network. We call a graph \emph{\anonvar-anonymous} if for every node in the
graph there exist at least $k$ other nodes that share at least
$\ell$ of its neighbors. We consider two combinatorial problems
arising from this notion of anonymity in graphs. More
specifically, given an input graph we ask for the minimum number
of edges to be added so that the graph becomes
\emph{\anonvar-anonymous}. We define two variants of this
minimization problem and study their properties. We show that for
certain values of $k$ and $\ell$ the problems are polynomial-time solvable,
while for others they become NP-hard. Approximation algorithms for
the latter cases are also given.
\end{abstract}

\pagenumbering{arabic}

\newcommand{\DEF}[1]{{\em #1\/}}

\newcommand\chic{\chi_c}
\newcommand\C{\hbox{${\cal C}$}}
\newcommand{\RR}{\mbox{$\mathbb R$}}
\newcommand{\NN}{\mbox{$\mathbb N$}}
\newcommand{\ZZ}{\mbox{$\mathbb Z$}}
\newcommand{\eopf}{\raisebox{0.8ex}{\framebox{}}}
\newcommand{\dist}{\hbox{\rm d}}
\renewcommand\a{\alpha}
\renewcommand\b{\beta}
\renewcommand\c{\gamma}
\renewcommand\d{\delta}
\newcommand\D{\Delta}
\newcommand{\directedchi}{\mbox{$\vec{\chi}$}}
\newcommand{\directedE}{\mbox{$\vec{E}$}}
\newcommand{\directedG}{\mbox{$\vec{G}$}}
\newcommand{\directedK}{\mbox{$\vec{K}$}}

\section{Introduction}\label{introduction}

The popularity of online communities and social networks
in recent years has motivated research on social-network analysis.
Though these studies are useful in uncovering
the underpinnings of human social behavior, they also raise
privacy concerns for the individuals involved.

A social network is usually represented as a graph, where nodes
correspond to individuals and edges capture relationships between
these individuals. For example, in LinkedIn, an online network of
professionals, every link between two users specifies a
professional relationship between them. In Facebook and Orkut
links correspond to friendships. There are online communities that
permit any user to access the information of every node in the graph
and view its neighbors. However, many communities are
increasingly restricting access to the personal information of other users. For example, in LinkedIn, a user can only see the profiles of his own friends and their connections.

In this paper, we consider a scenario where the owner of a
social network would like to release the underlying
graph of interactions for social-network analysis purposes, while preserving the privacy of its users. More specifically,
the private information to be protected
is the mapping of nodes to real-world entities and interconnections amongst them.
Therefore, we design an anonymization framework that tries to hide the identity of nodes by creating groups of nodes that look similar by virtue of sharing many of the same neighbors. We call such nodes {\em anonymized}. Our goal is to anonymize all nodes of the graph by introducing minimal changes to the overall graph structure. In this way we can guarantee that the anonymized graph is still useful for social-network analysis purposes.

Recently, Backstrom et. al.~\cite{backstrom} have shown that the most
simple graph-anonymization technique that removes the identity of
each node in the graph, replacing it with a random identification
number instead, is not adequate for preserving the privacy of
nodes. Specifically, they show that in such an anonymized network, there exists an adversary
who can identify target individuals and the link structure between them. However, the problem of designing anonymization methods against such
adversaries is not addressed in~\cite{backstrom}.

Following the work of~\cite{backstrom}, Hay et. al.~\cite{miklau}
have very recently given a definition of graph anonymity: a graph
is $k$-anonymous if every node shares the same neighborhood
structure with at least $k-1$ other nodes. The definition is
recursive, and has some nice properties studied in~\cite{miklau}.
However, the focus of~\cite{miklau} is mostly on the properties of
the definitions rather than on algorithms to achieve the anonymity requirements.

Motivated by ~\cite{backstrom} and~\cite{miklau}, Zhou and
Pei~\cite{zhou08preserving} consider the following definition of
anonymity in graphs: a graph is $k$-anonymous if for every node
there exist at least $k-1$ other nodes that share isomorphic
$1$-neighborhoods. They consider the problem of minimum
graph-modifications (in terms of edge additions) that would lead
to a graph satisfying the anonymity requirement. Although this
definition is interesting, the algorithm presented
in~\cite{zhou08preserving} is not supported by theoretical
analysis. Further, if the anonymity definition is extended to consider the neighborhood structure beyond just the immediate $1$-neighborhood of each node, algorithmic techniques quickly become infeasible.

Despite the fact that privacy concerns in releasing social-network data have
been pinpointed, there is no agreement on the definition of
privacy or anonymity that should be used for such data. In this paper, we try to move
this line of research one step forward by proposing a new definition
of graph anonymity that is inline to a certain extent with
the definitions provided in~\cite{miklau}. Our definition of
anonymity is in a sense less strict than the one proposed
in~\cite{zhou08preserving}. However, we consider it to be natural,
intuitive and more amenable to theoretical analysis. 

Intuitively our definition aims to protect an individual from an adversary who knows some subset of the individual's neighbors in the graph. After anonymization, the hope is that the adversary can no longer identify the target individual because several other nodes in the graph will also share this subset of neighbors. Further, during anonymization, the identifying subset of neighbors themselves will become distorted and harder for the adversary to identify.

\vspace{0.1in}
\noindent{\bf The Problem:} We define a graph to be
{\anonvar}-anonymous if for every node $u$ in the graph there
exist at least $k$ other nodes that share at least $\ell$ of their
neighbors with $u$. In order to meet this anonymity requirement one could transform
any graph into a complete graph. For a graph consisting of $n$
nodes this would mean that every node would share $n-2$ neighbors
with each of the $n-1$ other nodes. Although such an anonymization
would preserve privacy, it would make the anonymized graph useless
for any study. For this reason we impose the additional
requirement that the minimum number of such edge additions should
be made. The aim is to preserve the utility of the original graph,
while at the same time satisfying the {\anonvar}-anonymity
constraint.

Given $k$ and $\ell$ we formally define two variants of the
\emph{graph-anonymization} problem that ask for the minimum number
of edge additions to be made so that the resulting graph is
{\anonvar}-anonymous. We show that for certain values of $k$ and
$\ell$ the problems are polynomial-time solvable, while for others
they are NP-hard. We also present simple and intuitive
approximation algorithms for these hard instances. To summarize our contributions:

\squishlist
\item We propose a new definition of graph anonymity building on
previously proposed definitions.
\item We provide the first formal
algorithmic treatment of the graph-anonymization problem.
\squishend

\vspace{0.15in}
Besides graph anonymization, the combinatorial problems we study
here may also arise in other domains, e.g., graph reliability.
We therefore believe that the problem definitions and
algorithms we present are of independent interest.

\vspace{0.1in}
\noindent{\bf Roadmap:} The rest of the paper is organized as
follows. In Section~\ref{sec:related} we summarize the related
work. Section~\ref{sec:definitions} gives the necessary notation
and definitions. Algorithms and hardness results for different
instances of the {\anonvar}-anonymization problem are given in
Sections~\ref{2_1},~\ref{6_1to7_1},~\ref{K_1} and~\ref{K_L}. We
conclude in Section~\ref{sec:conclusions}.

\section{Related Work}\label{sec:related}

As mentioned in the Introduction, there has been some prior work
on privacy-preserving releases of social-network graphs. The authors
in~\cite{backstrom} show that the naive approach of simply masking
usernames is not sufficient anonymization. In particular, they
show that, if an adversary is given the chance to create as few as $\Theta(\log(n))$
new accounts in the network, prior to its release, then he can efficiently recover the
structure of connections between any $\Theta(\log^2 (n))$ nodes
chosen apriori. He can do so by identifying the new accounts that he inserted in
to the network. The focus of~\cite{backstrom} is on revealing the
power of such adversaries and not on devising methods to protect against them.

In~\cite{miklau} the authors experimentally evaluate
how much background information about the neighborhood of an individual would be sufficient for an adversary to uniquely identify that individual in a naively anonymized graph. Additionally, a new recursive definition of graph anonymity is given. The definition says that a graph is
$k$-anonymous if for every structure query there exist $k$ nodes
that satisfy it. The definition
is constructed for a certain class of structure queries
that query the neighborhood structure of the nodes.
Our definition of anonymity is inspired by~\cite{miklau}, however
it is substantially different. Moreover, the focus of our work is
on the combinatorial problems arising from our anonymity
definition.

Very recently, the authors of~\cite{zhou08preserving} consider yet
another definition of graph anonymity; a graph is $k$-anonymous if
for every node there exist at least $k-1$ other nodes that share
isomorphic $1$-neighborhoods. This definition of anonymity in
graphs is different from ours. In a sense it is a more strict one.
Moreover, though the algorithm presented
in~\cite{zhou08preserving} seems to work well in practice, no
theoretical analysis of its performance is presented. Finally,
extending the privacy definition to more than just the
$1$-neighborhood of nodes causes the algorithms
of~\cite{zhou08preserving} to quickly become infeasible.

The problem of protecting sensitive links between individuals in
an anonymized social network is considered in~\cite{Zheleva:07}.
Simple edge-deletion and node-merging algorithms are proposed to
reduce the risk of sensitive link disclosure. This work is
different from ours in that we are primarily interested in
protecting the identity of the individuals while
in~\cite{Zheleva:07} the emphasis is on protecting the types of links
associated with individuals. Also, the combinatorial problems that
we need to solve in our framework are very different from the set
of problems discussed in~\cite{Zheleva:07}.

In~\cite{Frikken:06} the authors study the problem of assembling
pieces of a graph owned by different parties privately. They
propose a set of cryptographic protocols that allow a group of
authorities to jointly reconstruct a graph without revealing the
identity of the nodes. The graph thus constructed is isomorphic to
a perturbed version of the original graph. The perturbation
consists of addition and or deletion of nodes and or edges. Unlike
that work, we try to anonymize a single graph by modifying it as
little as possible. Moreover, our methods are purely combinatorial
and no cryptographic protocols are involved.

Korolova et. al.~\cite{korolova} investigate an attack where an
adversary strategically subverts user accounts. He then uses the
online interface provided by the social network to gain access to
local neighborhoods and to piece them together to form a global picture. The authors provide
recommendations on what the lookahead of a social network should
be to render such attacks infeasible. This work does not
consider an anonymized release of the entire network graph and is thus different from ours.

Besides graphs, there has been considerable prior work on
anonymizing traditional relational data sets. The line of work on
$k$-anonymity found in
\cite{aggarwal,gehrke,lefevre,meyerson,sweeney,tclose} aims to
minimally suppress or generalize public attributes of individuals
in a database in such a way that every individual (identifiable by
his public attributes) is hidden in a group of size at least $k$.
Our notion of graph anonymity draws inspiration from this.

Apart from suppression or generalization techniques, perturbation
techniques have also been used to anonymize relational data sets
in \cite{dilys,haritsa,srikant}. Perturbation-based approaches for
graph anonymization are also considered in~\cite{miklau,xintao}; in that
case edges are randomly inserted or deleted to anonymize the graph. We do not consider perturbation-based approaches in this paper.

\section{Preliminaries}\label{sec:definitions}

In this section we formalize our definition of graph
anonymity and introduce two natural optimization problems
that arise from it.

Throughout the paper we assume that the social-network graph is
simple, \textit{i.e.}, it is undirected, unweighted, and contains
no self-loops or multi-edges. This is an important category of
graphs to study; most of the aforementioned social
networks (Facebook, LinkedIn, Orkut) allow only bidirectional
links and are thus instances of such simple graphs. We assume that
the actual identifiers of individual nodes are removed prior to
further anonymization. Our definition for graph anonymity is
inspired by the notion of $k$-anonymity for relational data
wherein each person, identifiable by his public attributes, is
required to be hidden in a group of size $k$. In the case of a
social-network graph, the publicly-known attributes of a user
would be (a subset of) his connections (and interconnections
amongst them) within the graph.

Consider a simple unlabelled graph and an adversary who knows that
a target individual and some number of his friends form a clique.
In the released graph, the adversary could look for such cliques
to narrow down the set of nodes that might correspond to the
target individual. The goal of an anonymization scheme is to
prevent such an adversary from uniquely identifying the individual
and his remaining connections in the anonymized graph.

We achieve this by introducing an anonymity property that requires
that for every node in the graph, some subset of its neighbors
should be shared by other nodes. In this way, an adversary
who knows some subset of the neighbors of a target individual and can even pinpoint them in the graph, will not be able to distinguish the target individual from other nodes in the network that share this subset of neighbors. Further, in the
process of anonymization, the identifying subset of neighbors
itself becomes distorted and harder for the adversary to pinpoint.
More formally we define the {\anonvar}-anonymity property as
follows.

\begin{definition}[\anonvar-anonymity]
A graph $G=(V,E)$ is {\em \anonvar-anonymous} if for each vertex
$v\in V$, there exists a set of vertices $U\subseteq V$ not
containing $v$ such that $|U|\geq k$ and for each $u\in U$ the
vertices $u$ and $v$ share at least \secondvar\ neighbors.
\end{definition}

\begin{example}
A clique of $n$ nodes is $(n-1,n-2)$-anonymous.
\end{example}

To demonstrate the kinds of attacks we hope to protect against, we give another example.

\begin{figure}
\begin{center}
\subfigure[Input graph $G$]{\includegraphics[scale=0.75,angle=270]{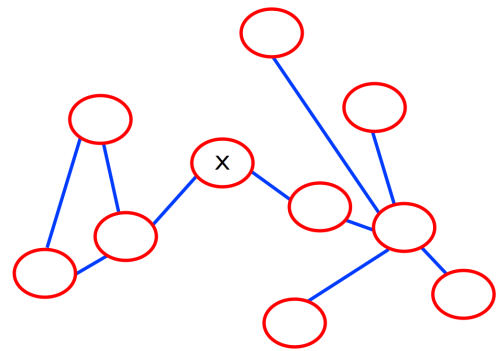}\label{41examplea}}
\subfigure[(4,1)-anonymous transformation of $G$]{\includegraphics[scale=0.75,angle=270]{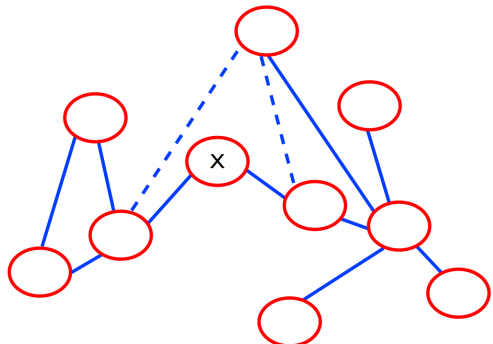}
\label{41exampleb}}
\caption{In Figure~\ref{41examplea} an adversary can identify Alice as the node marked X. Figure~\ref{41exampleb} is a (4,1)-anonymous transformation of the graph.}\vspace{-0.15in}
\end{center}
\end{figure}

\begin{example}
Consider the graph in Figure~\ref{41examplea}. Suppose an adversary knows that Alice is in this graph and that Alice is connected to a friend who is part of a triangle. There is only one such node in the graph and hence the adversary will be able to determine that the node marked X in the graph uniquely corresponds to Alice. From this he may be able to further infer the identities of Alice's neighbors and their neighbors as well. Now if the edges shown in dotted lines in Figure~\ref{41exampleb} are added to this graph, the resulting graph is $(4,1)$-anonymous. In this new graph, Alice is no longer the only node connected to a node of a triangle. Further, there is no longer only one triangle in the graph. 
\end{example}

Given an input graph $G=(V,E)$ with $n$ nodes, and integers $k$ and
$\ell$, our goal is thus to transform the graph into a
{\anonvar}-anonymous graph. We focus on transformations that
allow only additions of edges to the original graph
In order for the anonymized graph to remain useful for social-network (or
other) studies, we need to ensure that the transformed graph is as
close as possible to the original graph. We achieve this by
requiring that a minimum number of edges should be added to $G$ so that
the $(k,\ell)$-anonymity property holds. This leads us to the
following two variants of the {\anonvar}-anonymization problem.

\begin{problem}[Weak \anonvar-anonymization]\label{problem:weak}
Given a graph $G=(V,E)$ and integers $k$ and $\ell$, find the
minimum number of edges that need to be added to $E$, to obtain a
graph $G' = (V, E')$ that is \anonvar-anonymous.
\end{problem}

The following example illustrates the
weak-anonymization problem.

\begin{figure}
\begin{center}
\subfigure[Input graph $G$] {\includegraphics[scale =
0.45]{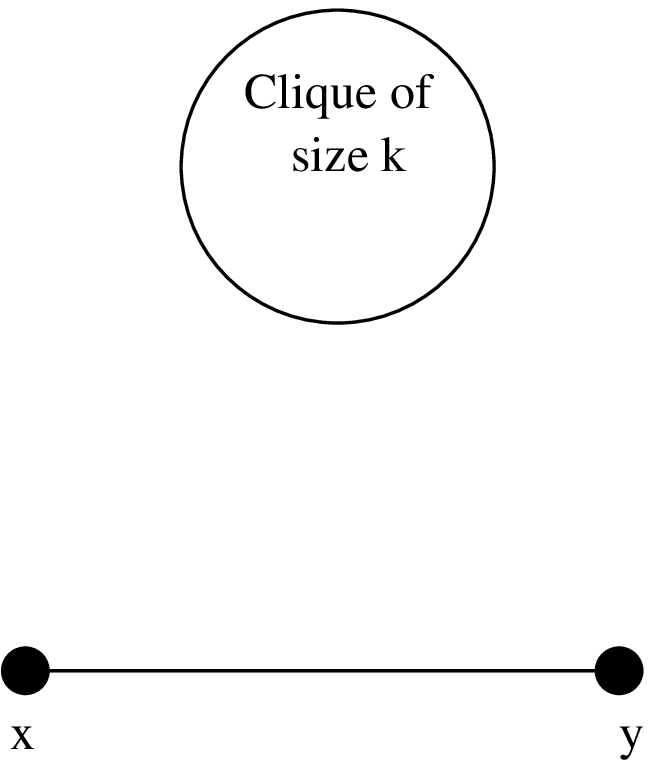}\label{fig:input_graph}} \hspace{1cm}
\subfigure[Weakly $(k-1,1)$ - anonymized graph $G'$]
{\includegraphics[scale =
0.45]{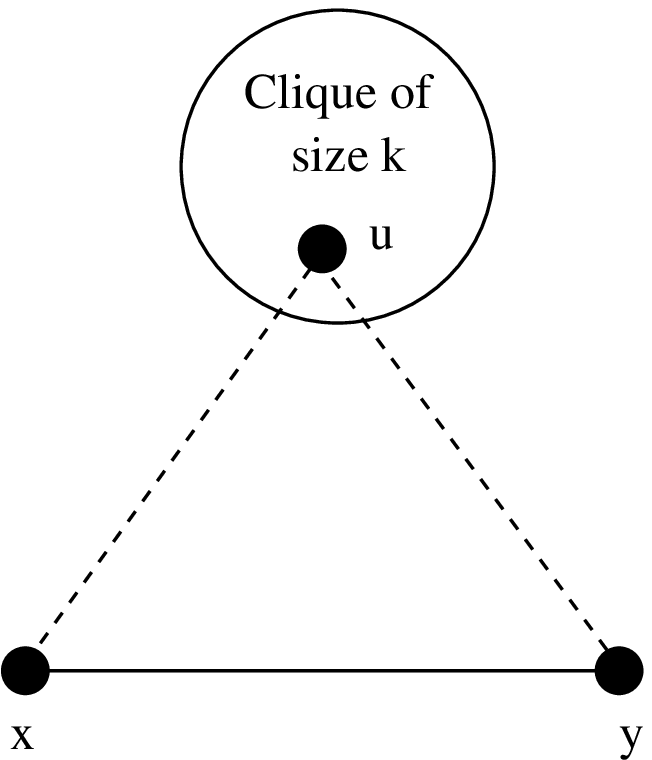}\label{fig:weak_anonymized}}\caption{Illustrative
example of the difference between \emph{weak} and \emph{strong}
anonymity. }\label{figure:example}
\end{center}
\end{figure}

\begin{example}\label{example:weak}
Consider the input graph $G$ of Figure~\ref{fig:input_graph}. The
graph consists of a clique of size $k$ and $2$ nodes $x$ and $y$
connected by an edge. The nodes in the clique are all $(k-1,
k-2)$-anonymous. However, the existence of $x$ and $y$ prevents
$G$ from being fully $(k-1, 1)$-anonymous.

Assume now that we connect both $x$ and $y$ to a single node $u$ of the clique.
In this way, we construct graph $G'$ shown in
Figure~\ref{fig:weak_anonymized}. Obviously, $G'$ is
$(k-1,1)$-anonymous; all the nodes in $G'$ (including $x$ and $y$) have
$k-1$ other nodes that share at least one of their neighbors. For
$x$ and $y$, this neighbor is node $u$.
\end{example}

The problem in the above example is that graph $G'$ satisfies the
$(k-1,1)$-anonymity requirement, however, the anonymity of nodes
$x$ and $y$ is achieved via node $u$ that was not a part of their
initial set of neighbors in $G$. Thus, the goal of having many
other nodes sharing the original neighborhood structure of $x$ or
$y$ is not necessarily achieved unless we place additional
requirements on the anonymization procedure. To this end we
introduce the problem of \emph{strong anonymization}. Strong
anonymity places additional restrictions on how anonymity can be achieved
and provides better privacy.

\begin{definition}[Strong
\anonvar-transformation]\label{dfn:strong_anonymity} Consider
graphs $G=(V,E)$ and $G'=(V,E')$, so that $E\subseteq E'$ and $G'$
is {\anonvar}-anonymous. For fixed $k$ and $\ell$, we say that
$G'$ is a \emph{strongly-anonymized transformation} of $G$, if for
every vertex $v\in V$, there exists a set of vertices $U\subseteq
V$ not containing $v$ such that $|U|\geq k$ and for each $u\in U$,
$|N_G(v) \cap N_{G'}(u)|\geq$ \secondvar. Here $N_G(v)$ is the set
of neighbors of $v$ in $G$, and $N_{G'}(u)$ is the set of
neighbors of $u$ in $G'$.
\end{definition}

Therefore, if a graph $G'$ is a strong {\anonvar}-transformation
of graph $G$, then each vertex in $G'$ is required to have $k$
other vertices sharing at least \secondvar\ of its {\em original}
neighbors in $G$. For this to be possible, every vertex must have at least $\ell$ neighbors in the original graph $G$ to begin with.

\begin{example}\label{example:strong}
Consider again the graph $G$ of Figure~\ref{fig:input_graph} and
its transformation to graph $G'$ shown in
Figure~\ref{fig:weak_anonymized}. In Example~\ref{example:weak} we
showed that graph $G'$ is $(k-1,1)$-anonymous in the weak sense.
However, in order to get a strong $(k-1,1)$-transformation
of $G$, we would have to connect each of the nodes $x$ and $y$ to $k-1$ other nodes from the clique.
\end{example}

The definition of a strong {\anonvar}-transformation gives rise to
the following \emph{strong {\anonvar}-anonymization} problem.

\begin{problem}[Strong \anonvar-anonymization]\label{problem:strong}
Given a graph $G=(V,E)$ and integers $k$ and $\ell$, find the
minimum number of edges that need to be added to $E$, to obtain
graph $G' = (V, E')$ that is a strong \anonvar-transformation of
$G$.
\end{problem}

Obviously achieving strong anonymity would require the addition of
a larger number of edges than weak anonymity. This statement is formalized as follows.

\begin{proposition}
Consider input graph $G=(V,E)$ and integers $k$ and $\ell$. Let
$G'=(V,E')$ be the $(k,\ell)$-anonymous graph that is the optimal
solution for Problem~\ref{problem:weak}, and $G''=(V,E'')$ be the
$(k,\ell)$-anonymous graph that is the optimal solution for
Problem~\ref{problem:strong}. Then it holds that $|E''|\geq |E'|$.
\end{proposition}

The notion of {\anonvar}-anonymity is strongly related to the
immediate neighbors of a node in the graph, and how these are
shared with other nodes. Therefore, for every node $u$ it is
important to know the nodes that are reachable from $u$ via a path
of length exactly $2$. Given its importance, we define the notion
of $2$-neighborhood of a node as follows.

\begin{definition}[$2$-neighborhood]
Given a graph $G=(V,E)$ and a node $v\in V$ we define the $2$-neighborhood of $v$ to be the set of all nodes in $G$ that are
reachable from $v$ via paths of length exactly $2$.
\end{definition}

We also define two more terms that will be used in the rest of the paper.

\begin{definition}[Residual Anonymity]\label{dfn:residual}
Consider a graph $G=(V,E)$ that we would like to make
$(k,\ell)$-anonymous. Consider any node $v \in V$ and suppose that
$k'$ other nodes in the graph share at least $\ell$ of $v$'s
neighbors. Then, we define the residual anonymity of $v$ to be
$r(v) = {\text max}\{k-k', 0\}$. The residual anonymity of a graph
$G=(V,E)$ is defined to be $r(G) = \sum_{v \in V} r(v)$.
\end{definition}

We define the concept of a deficient node for nodes that are not $(k,\ell)$-anonymous.

\begin{definition}[Deficient Node]
A node $v$ is deficient if $r(v) > 0$.
\end{definition}

It is the deficient nodes that we need to take care of in order to
anonymize a graph. With these definitions in hand, we are now ready
to proceed to the technical results of the paper.

\section{$(2,1)$-anonymization}\label{2_1}

In this section we provide polynomial-time algorithms for the weak
and strong $(2,1)$-anonymization problems. First, it is easy to
see that there is a simple characterization of $(2,1)$-anonymous
graphs. This fact is captured in the following proposition.

\begin{proposition}\label{prop:characterization}
A graph $G=(V,E)$ is $(2,1)$-anonymous if and only if each vertex
$u\in V$ is (a) part of a triangle, (b) adjacent to a vertex of
degree at least 3, or (c) is the middle vertex in a path of 5 vertices.
\end{proposition}

The main idea of the algorithms that we develop for
$(2,1)$-anonymization is that they add the minimum number of edges
so that every vertex of the resulting graph satisfies one of the
conditions of Proposition~\ref{prop:characterization}. Both
algorithms proceed in two phases: the \emph{deficit-assignment}
and the \emph{deficit-matching} phase. The deficit assignment
requires a linear scan of the graph in which deficits are assigned
to vertices. Roughly speaking, a deficit of $1$ signifies that the vertex needs to be
connected to another vertex of non-zero deficit by the addition of an extra edge. 
This added edge ensures that the $(2,1)$-anonymity requirement for the vertex or its
neighbors will be satisfied. Once the deficits are assigned to vertices
the algorithms proceed to the actual addition of edges. The edges
are added by taking into account the deficits of all vertices. For
example, two vertices both of deficit $1$ can be connected by the addition of a
single edge (if they are not already neighbors and are not
isolated). In this way, a single edge accommodates a total
deficit of $2$. The minimum number of edges to be added can be found via a matching of the vertices with deficits. The matching consists of edges that are not already in the graph. A perfect matching is the matching that satisfies all the deficits. In the case of weak anonymization, this matching can be found in linear time by randomly pairing up non-adjacent vertices with deficits. For strong anonymization, it needs to be explicitly computed by solving the maximum-matching problem over edges that are not already in the graph.

Another key point in the development of our algorithms is that in
order to assign deficits it suffices to explore only vertices that
are within a distance $4$ from some leaf vertex or from a vertex
of degree $2$. Any other vertex can be shown to satisfy the conditions of Proposition~\ref{prop:characterization}.
Finally, it only requires a case analysis to show that our
algorithms optimally assign deficits to vertices,
independently of the order in which they traverse the vertices of
the input graph during the first phase. For lack of space we only
give a sketch of the algorithms and proofs in this section.

\subsection{Linear-time weak $(2,1)$-anonymi-zation}
As we have already mentioned our algorithm for the weak
$(2,1)$-anonymization problem has two phases (1) deficit
assignment and (2) deficit matching\footnote{Recall that a node $u$ is assigned deficit $i$ if $i$ edges need to be added between other non-zero deficit vertices and $u$ in order to satisfy the anonymity requirements of $u$ or $u$'s neighbors.}

\vspace{0.1in}
\noindent{\bf Deficit Assignment:} First assume that the input graph has no
isolated vertices -- we will show how to deal with isolated
vertices later. For the deficit-assignment phase, the algorithm
starts with an {\em unmarked} vertex of
degree $1$ or $2$ and explores vertices within a distance $4$ of it. Deficits are assigned as follows:

\begin{itemize}
\item For an isolated edge $uv$, we assign deficit $1$ to $u$ and
deficit $1$ to $v$; it may be that both edges will be added at $u$.

\item For an isolated path $uvw$, we assign deficit $1$ to $v$.

\item For an isolated path $uvwx$, we assign deficit $1$ to $v$ and
deficit $1$ to $w$.

\item For a subgraph consisting of a path $uvw$ with adjacent
vertices attached to $w$, we assign deficit $1$ to $v$.

\item For a component $uvX_i$ with vertex $u$ having degree one
with vertex $v$ connected to a set of vertices $X_i$ such that
each $x\in X_i$ has degree $1$ (and no other vertices) assign
deficit $1$ to $v$. This component corresponds to an isolated star
centered at $v$.

\item For a component consisting of a square $uvwx$ (isolated
square), we assign deficit $1$ to $u$ and deficit $1$ to $w$; it may
be that the two edges will be added at $u$ and $v$, or that $u$
and $w$ will be joined.

\item For a subgraph consisting of a square $uvwx$ with edges (one
or more) $ux_i$ coming out of the square, we assign deficit $1$ to
$v$.

\item For a subgraph consisting of squares $uv_1wx_1$,
$uv_2wx_2$, $\ldots$, $uv_jwx_j$, we assign deficit
$1$ to one of the $v_i$'s.

\item Finally, for a subgraph consisting of a vertex $u$ adjacent
to vertices $x_i$ of degree $1$ and to a vertex $y$ of degree $2$,
assign deficit $1$ to $y$.

\end{itemize}

All the vertices that are visited in this process are {\em marked}
(that is the assigned deficits cover all marked vertices) and the
deficit-assignment process repeats starting with the next
unmarked vertex until no more unmarked
vertices of degree $1$ or $2$ remain.

\vspace{0.1in}
\noindent{\bf Deficit Matching:} If the
number of vertices with deficit $1$ is $2m$, and $2m\geq 4$ or
$2m=2$ -- in some case other than an isolated edge $uv$ -- then, we
need to find any perfect matching amongst these
vertices to find the edges to add. 
The matching of deficits can be done in linear time since any
(random) pairing of non-adjacent vertices with non-zero deficits suffices. In this case we add $m$ extra edges. 
If the number of vertices with deficit $1$ is $2m+1$, then all but one of
these vertices can be matched, and a single edge needs to be added
to the remaining vertex, connecting it to some vertex of degree
at least $2$. This results in a total of $m+1$ extra edges. 
There are, however, some special cases that we need to take care of first.

\vspace{0.1in}
\noindent{\bf Special Cases:} Before finding the perfect matching we match all
isolated edges to each other. This is because the isolated edges
need to be connected in a special way to take care of the deficits
at the two ends. For a pair of isolated edges $uv$ and $u'v'$, we add
the edges $uu'$ and $vu'$ (we treat the two deficits of $1$ at $u$
and $v$ as being concentrated at $u$). In the end we may be left
with a single isolated edge $uv$. In this case, two edges need to be added
and we can connect them to any other vertex in the graph
forming a triangle.  Similarly, in the case where
the remainder is an isolated star centered at $v$ with vertices
$x_i$ of degree one, it is enough to add a single edge to
connect vertices $x_j$ and $x_{j'}$ of the star.

\vspace{0.1in}
\noindent{\bf Isolated Vertices:} It remains to take care of isolated
vertices. For this we consider a set of six isolated vertices
$u,v,w,u',v',w'$  and we connect them with edges $uv, uw, uu',$
$u'v', u'w'$. These five edges can take care of the six isolated
vertices. In general, the vertices with deficit $1$ can be attached
to isolated vertices first, with two exceptions to be considered
next. When we have an isolated edge $xy$, one of the two deficits
of 1 can be satisfied by connecting $x$ to an isolated vertex, but
the other one can also be satisfied by connecting $x$ to an
isolated vertex $u$ if $u$ is also made adjacent to two other
isolated vertices $v$ and $w$ to obtain the above mentioned
component. Similarly if $x$ is only adjacent to vertices $y_i$ of
degree $1$, then the deficit $1$ at $x$ can only be matched to an
isolated $u$ if $u$ is also made adjacent to two other isolated
vertices $v$ and $w$. In the end we will be left with fewer than
six isolated vertices which each need one edge. These can be
connected to any vertex in the graph of degree at least $2$. The
optimality follows because a tree on $5$ vertices is optimal saving.

\begin{theorem}
The above algorithm solves optimally the weak
$(2,1)$-anonymization problem in linear time.
\end{theorem}

{\em Proof Sketch:} It requires a case analysis (that we omit for lack of space) to show
that the deficit-assignment scheme we described above is complete and optimal
and that the total deficit assigned is independent of the order in
which the vertices of the graph are traversed. Since we find a
perfect matching, we satisfy these deficits with as few edges as
possible, hence, the optimality of the algorithm.

It is also easy to see that the deficit-assignment takes time
linear with respect to the number of edges in the graph: first we
only consider vertices of degree one or two as starting points. For every such vertex we only have to explore all
vertices within a distance $4$. This is because any other vertex can be seen to satisfy one of the conditions of Proposition~\ref{prop:characterization}. After each iteration of the deficit assignment, we
mark all the vertices that have been visited in this process as
marked (that is the assigned deficits cover all visited
vertices). The deficit-assignment process continues starting
with the next unmarked vertex of degree $1$ or $2$. The scanning of
the algorithm requires only linear time with respect to the number
of edges in the graph since every traversed edge connects only
marked endpoints and thus no edge needs to be traversed more than once
by the algorithm.

The deficit-matching phase is also linear since it only requires to
find any (random) matching between non-adjacent deficits.

\subsection{Polynomial-time strong $(2,1)$-anonymization}

The algorithm for solving the strong $(2,1)$-anonymization problem
is very similar to the one presented in the previous section, so
we only briefly discuss it here. For brevity we avoid mentioning
various special cases that are similar to the weak-anonymization
problem. The first key difference is that for strong
$(2,1)$-anonymization we need to develop a different
deficit-assignment scheme. Although the actual structures we have
to consider for assigning the deficits are the same we need to
assign different deficits to different vertices so that we satisfy
the strong anonymity requirement. This is because an edge added at
a vertex with assigned deficit can only help the original
neighbors of the vertex, and not the vertex itself. The second
difference is that in the deficit-matching phase we need to
actually solve a maximum-matching problem; not every random pairing of non-adjacent vertices with assigned deficit is a valid solution.

In strong $(2,1)$-anonymization we first have to assume that there
are no isolated vertices in the input graph $G$; otherwise strong
$(2,1)$-anonymity is not achievable for these vertices.

\vspace{0.1in}
\noindent{\bf Deficit Assignment:} For the deficit-assignment step, the
algorithm starts with an unmarked vertex in the input graph
with degree $1$ or $2$ and assigns deficits as follows:

\begin{itemize}
\item For an isolated edge $uv$, assign deficit of $2$ at each end.

\item For an isolated path $uvwx$, put deficit $1$ at $v$ and at
$w$.

\item For an isolated square $uvwx$, put deficit $1$ at $u$ and
$v$.

\item If such a square has edges already coming out of $v$, put
just deficit $1$ at $u$.

\item If multiple squares $uv_iwx_i$ all start from vertex $u$,
then assign deficit $1$ to one of the $v_i$'s.

\item For a path $uvw$, put deficit $1$ at each of the $3$ vertices.

\item For a vertex of degree at least $3$ attached to vertices of
degree $1$, put two deficits of $1$ at degree $1$ vertices.

\item If a path starts $uvwx$, with $x$ of degree at least $2$, put
deficit $1$ at $v$ and $1$ at $w$.

\item If in addition $w$ has other edges coming out of it, put
deficit $1$ just at $v$. Otherwise if in addition only $v$ has other
edges coming out of it that join to a vertex of degree $1$, put
deficit $1$ just at $w$.

\end{itemize}

All vertices that are visited in the process are marked, and the
algorithm proceeds with the next unmarked vertex until there are
no unmarked vertices left.

\vspace{0.1in}
\noindent{\bf Deficit Matching:} For solving the strong $(2,1)$ -
anonymization problem exactly we need to solve a maximum-matching
problem between the nodes with deficits. This can be done in polynomial
time~(\cite{papadimitrioucombinatorial}). Note, that in the weak
$(2,1)$-anonymization problem \emph{any} random pairing of
non-adjacent nodes with deficits was sufficient, allowing for a linear-time matching phase.
This was because with the exception of isolated edges and isolated paths of length $4$, there was no case in which two vertices of non-zero deficit could be adjacent. This is not the case in the strong anonymization problem, and here a maximum-matching problem needs to be solved over edges that are not already in the graph.

A linear-time deficit-matching algorithm with a small
additive error can also be developed. This is summarized in the
following theorem.

\begin{theorem}
The strong $(2,1)$-anonymization problem can be approximated in
linear time within an additive error of 2, and can be solved
exactly in polynomial time.
\end{theorem}

{\em Proof Sketch:} It requires again a case analysis to
show that the deficit-assignment scheme is optimal and independent
of the order in which we traverse the vertices.

Now, if all deficits add up to $m$, they can easily be paired
using a greedy linear-time matching algorithm. However, 
the last $2$ deficits may be assigned to adjacent
vertices. So instead of adding $\lceil m/2\rceil$ edges, we may
add $\lceil m/2\rceil+2$, for an additive error of 2. If instead
we use a maximum-matching algorithm to match as many deficits
as possible and satisfy the unmatched deficits individually, the
problem can be solved optimally in polynomial time.

\vspace{-0.15in}
\section{From $(6,1)$ to $(7,1)$-anonymity}\label{6_1to7_1}
We show here that given a graph that is already (6,1)-anonymous,
it is NP-hard to find the minimal number of edges that need to be
added to make it either weakly or strongly (7,1)-anonymous. This result provides insight into the complexity of the anonymization problem, showing that it is hard to achieve anonymity even incrementally. The
result follows from a reduction from the {\em 1-in-3
satisfiability} problem. An instance of 1-in-3 satisfiability
consists of triples of Boolean variables $(x,y,z)$ to be assigned
values 0 or 1 in such a way that each triple contains one 1 and
two 0s. This problem was shown to be NP-complete by
Schaefer~\cite{schaefer78complexity}. We first show that even a
restricted form of the 1-in-3 satisfiability problem is
NP-complete.

\begin{lemma}\label{lemma1}
The 1-in-3 satisfiability problem is NP-complete even if each variable
occurs in exactly 3 triples, no two triples share more than one variable, and
the total number of triples is even.
\end{lemma}

\begin{proof}
We prove this by taking an arbitrary instance of the 1-in-3 satisfiability problem and converting it to an instance satisfying the constraints of the above lemma. We start off by renaming multiple occurrences of a variable $x$ as $x_1$, $x_2$, and so on, so that by the end, each variable occurs in at most 1 triple and no two triples share more than one variable. We can then enforce the condition that each $x_i$ be equal to $x_{i+1}$ by inserting the triples
$(x_i,u,v)$, $(x_{i+1},u',v')$, $(u,u',w)$ and $(v,v',w)$.
This guarantees at most 3 occurrences of
each variable in triples. If a variable $y$ occurs in 2 triples, we may include
a triple $(y,z,t)$ introducing two new variables, so that at the end of this process each variable occurs in either
1 or 3 triples. Finally we make nine copies of the entire instance, each labeled $(i,j)$ with
$1\leq i,j\leq 3$, and equate the $z$s that have the same $i$ and also equate
the $t$s that have the same $j$. This guarantees that each variable appears
in exactly 3 triples. Making two copies of this instance guarantees that the
number of triples is even.
\end{proof}

\begin{theorem}
Suppose $G$ is $(6,1)$-anonymous. Finding the smallest set of
edges to add to $G$ to solve the weak or strong
$(7,1)$-anonymization problem is NP-hard. The same results hold
for going from $(k,1)$-anonymity to weak or strong $(k+1,
1)$-anonymity when $k \geq 6$.
\end{theorem}

\begin{proof}
We show this via a reduction from the 1-in-3 satisfiability
problem. We take an instance of the 1-in-3 satisfiability problem
satisfying the constraints of Lemma~\ref{lemma1}. We further
assume that the number of triples in this instance is a multiple
of 3, since if it is not a multiple of 3, it is easy to see that
there will be no satisfying assignment. Since we also assume that
the number of triples is even, the number of triples is in fact of
the form $6m$.

Taking this instance, we now form a cubic bipartite graph
$G=(U,V,E)$ by creating a vertex in $U$ for each triple and a
vertex in $V$ for each variable, with the two vertices connected
by an edge if the variable occurs in the triple. We add 5 new
neighbors of degree 1 to each vertex in $U$. Each of these added
neighbors and the vertices in $V$ are $(7,1)$-anonymous, but the
vertices in $U$ have only 6 vertices at distance 2, namely the 2
other neighbors of each of the 3 neighbors in $V$, giving
$(6,1)$-anonymity. We would like to increase the anonymity of
these vertices so that they are also $(7,1)$-anonymous. Note that
a solution to this anonymity problem has to consist of at least
$m$ edges. This is because the total residual anonymity of the graph is $6m$ and each new edge can reduce the residual anonymity by at most $6$. Now, if it were possible to select $2m$ vertices in $V$
that were adjacent to all the $6m$ vertices in $U$, we could
insert a perfect matching of $m$ edges between these $2m$ vertices
and simultaneously increase the anonymity of all the vertices in
$U$ by at least 1. This would correspond to a solution to the
1-in-3 satisfiability problem. Similarly, if there is a solution
to the anonymity problem that involves the addition of only $m$
edges, it must necessarily correspond to a solution to the 1-in-3
satisfiability problem. Thus a solution to the 1-in-3
satisfiability problem exists if and only if the solution to the
anonymity problem involves the addition of $m$ extra edges.

For $k\geq 6$, add $k-2$ nodes of degree 1 attached to each vertex
in $U$. Attach an additional node of degree $k-5$ to each vertex
in $U$. Attach the remaining $k-6$ neighbors of each such
additional node to a clique of size $k+2$. The result then follows
from the case of $k=6$.
\end{proof}

The complexity of minimally obtaining weak and strong
$(k,1)$-anonymous graphs remains open for $k=3,4,5,6$.

\section{{\an}-anonymization}\label{K_1}
We start our study for the $(k,1)$-anonymization problem by giving
two simple $O(k)$-approximation algorithms. We then show that the
approximation factor can be further improved to match a lower bound.
\subsection{${\text O}(k)$-approximation algorithms for $(k,1)$-anonymization}
Let $G=(V,E)$ be the input graph to the weak $(k,1)$-anonymization
problem. Consider the following simple iterative algorithm: at
every step $i$ add to graph $G_i$ ($G_1 = G$) a single edge
between a neighbor of a deficient node $u$ and a node that is not
already in the $2$-neighborhood of $u$ in $G_i$. If there are only
isolated deficient nodes in $G_i$, the algorithm directly connects
a deficient node to a node of a $(k+1)$-clique. If no such clique
exists, the algorithm creates it in a preprocessing step; $(k+1)$
randomly selected nodes are picked for this purpose. Repeat the
process until no deficient nodes remain. We call this algorithm
the {\weakaddition} algorithm. We show that {\weakaddition} is an
${\text O}(k)$-approximation algorithm for the weak
$(k,1)$-anonymization problem. This result is summarized in the
following theorem.

\begin{theorem}\label{thm:weak1}
{\weakaddition} gives a ${\text O}(k)$-approximation
for the weak $(k,1)$-anonymization problem. If the optimal
solution is of size $t$, {\weakaddition} adds at
most $4kt + k^2$ edges.
\end{theorem}

\begin{proof}
Let $R = \sum_{v\in V}r(v)$ be the residual anonymity (see
Definition~\ref{dfn:residual}) of graph $G = (V,E)$. Let
${\text{\sc Wa}}$ be the total number of edges added by the
{\weakaddition} algorithm. It holds that ${\text{\sc Wa}} \leq
R+k^2$. This is because at every step the algorithm adds one edge
that decreases the residual anonymity of the graph by at least
$1$. Therefore the algorithm adds at most $R$ edges. The
additional $k^2$ edges may be required to create a $(k+1)$-clique if such a clique does not exist.

Now assume that the optimal solution adds $t$ edges. Consider an
edge $uv$ of the optimal solution. This edge, at the time of its
addition, could have decreased the residual anonymity of the graph
by at most $4k$. This is because it could have decreased the
residual anonymity of each of $u$ and $v$ as well as the residual
anonymities of at most $k$ neighbors connected to $u$ and at most
$k$ neighbors connected to $v$ (if $u$ or $v$ had more than $k$
neighbors, then none of these neighbors would have been
deficient). Further, the edge $uv$ could have decreased the
residual anonymity of $u$ or $v$ by at most $k$, and the residual
anonymities of each of the $k$ neighbors of $u$ or each of the $k$
neighbors of $v$ by at most 1.

Thus, each edge of the optimal solution could have reduced the
residual anonymity of the graph by at most $4k$ at the time of its
addition. That is, $t \geq R/4k$.

Thus it is clear that $\text{\sc Wa}\leq 4kt +k^2$.
\end{proof}

For the strong $(k,1)$-anonymization problem we show that the
{\strongaddition} algorithm (very similar to {\weakaddition}), is
an ${\text O}(k)$-approximation. {\strongaddition} is also
iterative: in each iteration $i$ it considers graph $G_i$ and adds
one edge to it. The edge to be added is one that connects a
neighbor of a deficient node $u$ to a node that is not already in
the $2$-neighborhood of $u$. This process is repeated till no
deficient nodes remain. We can state the following for the
approximation ratio achieved by the {\strongaddition} algorithm.

\begin{theorem}
{\strongaddition} is a $2k$-approximation algorithm
for the strong $(k,1)$-anonymization problem.
\end{theorem}

\begin{proof}
As in the proof of Theorem~\ref{thm:weak1} consider input graph
$G=(V,E)$ with initial residual anonymity $R$. Every edge added by
the {\strongaddition} algorithm would reduce the residual
anonymity of the graph by at least $1$. Therefore, if the number
of edges added by the {\strongaddition} algorithm is $\text{\sc
Sa}$ we have that $\text{\sc Sa}\leq R$.

Suppose now that the optimal solution adds $t$ edges. An added
edge $uv$ decreases the residual anonymity of the graph by at most
$2k$. This is because the edge can decrease the residual anonymity
of only the {\em original} neighbors of $u$ and $v$ by at most $1$
each and there can be at most $2k$ such deficient neighbors. Thus
$t \geq R/2k$.

From the above we have that $\text{\sc Sa}\leq 2kt$.
\end{proof}

\subsection{$\Theta(\log n)$-approximation algorithms for $(k,1)$-anonymization}

We now provide two simple greedy algorithms for the weak and
strong \an-anonymization problems and show that they output
solutions that are ${\text O}(\log n)$-approximations to the
optimal. We then show that this is the best approximation factor
we can hope to achieve for arbitrary $k$.

We start by presenting {\weakgreedy} which is an ${\text O}(\log
n)$-approximation algorithm for the weak $(k,1)$-anonymization
problem. Consider input graph $G=(V,E)$ that has total
residual anonymity $R$.  The optimal solution to the problem
consists of a set of edges that together take care of all the
residual anonymity in the graph.

\begin{figure}[]
\begin{center}
{\includegraphics[scale =
0.35]{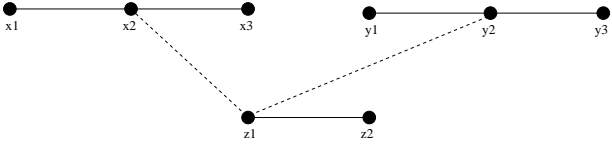}}\caption{Illustrative example of the
reinforcement between new edges in the weak-anonymization
problem.}\label{figure:reinforcement}
\end{center}
\end{figure}

We may be tempted to use a set-cover type solution: greedily
choose edges to add that maximally reduce the residual anonymity
of the graph at each step. However, such a greedy algorithm is not
so easy to analyze in the context of the weak-anonymization
problem. The difficulty in the analysis stems from the fact that
the new edges may \emph{reinforce} each other. That is, the
addition of an edge may bring about a greater reduction in the
residual anonymity of the graph in the presence of other added
edges. Consider, for example, the input graph $G$ shown in
Figure~\ref{figure:reinforcement}. Note that solid lines
correspond to the original edges in $G$. In this case, the
addition of edge $x2z1$ alone does not help in the anonymization
of node $y2$. (Neither does the addition of edge $y2z1$ in the
anonymization of $x2$). However, if edge $y2z1$ is already added
in the graph, then edge $x2z1$ helps in anonymizing node $y2$
as well.

We get around this peculiarity of our problem by greedily choosing
triplets of edges to add instead of singleton edges. Algorithm~1, called {\weakgreedy}, describes the procedure.

\begin{algorithm}[H]\label{alg1}
\caption{{\weakgreedy} for weak $(k,1)$-anonymization}
\begin{algorithmic}[1]
\STATE //Input: $k, G=(V,E)$
\STATE Randomly choose a node $w \in V$
\STATE Add up to ${k+1 \choose 2}$ edges to $E$ to form a $k+1$-clique at $w$
\STATE Compute $R = $ residual anonymity of $G$
\WHILE{$R > 0$}
\STATE Find triplet $uv, uw, vw$ that maximally decrease $R$
\STATE $E = E \cup \{uv\} \cup \{uw\} \cup \{vw\}$
\STATE Update $R$
\ENDWHILE
\end{algorithmic}
\end{algorithm}

\begin{theorem}
{\weakgreedy} is a polynomial-time nearly ${\text
O}(\log n)$-approximation algorithm for the weak \an-anony- mization
problem. If the optimal solution is of size $t$, the algorithm
adds $k^2+6t\log n$ edges.
\end{theorem}

\begin{proof}
Consider the optimal solution of $t$ edges. These $t$ edges
together take care of all the residual anonymity in the graph. We
can convert this solution to a solution of triplets that consists of at most $k^2 +
3t$ edges: first randomly choose a node $w$ and create a
$(k+1)$-clique amongst $w$ and $k$ other randomly chosen nodes.
Then, for each edge $uv$ of the optimal solution, add a triangle
$(uv, vw, uw)$ to the graph. The resulting graph will clearly
continue to be \an-anonymous. The $t$ triangles in conjunction
with the $(k+1)$-clique take care of all the residual anonymity in
the graph. Further, these triangles do not reinforce each other
because they are all connected to a node of degree $k$.

Going back to Algorithm~1, this means that once a $(k+1)$-clique has been added to
the graph, at each iteration of the algorithm, there must exist
some triangle with a vertex in the $(k+1)$-clique that reduces the
residual anonymity of the graph by a factor of at least $t$ (similar to the argument for the greedy set cover algorithm). And
since the algorithm greedily chooses triangles to add, the
residual anonymity of the graph will decrease by at least this
factor at each step. Since the residual anonymity of the graph can
be at most $kn < n^2$ to begin with, the algorithm will only
proceed for at most $r$ iterations till $(1 - 1/t)^r \leq  1/kn$.
This would mean that $r= {\text O}(t\log(kn)) = {\text O}(2t\log
n)$ and $3r =  {\text O}(6t\log n)$.
\end{proof}

The approximation algorithm for the strong \an-anony- mization
problem is simpler, since added edges cannot
reinforce each other --- an added edge can only help the original
neighbors of its two end points. Algorithm~2 gives the details of
the {\stronggreedy} algorithm.

\begin{algorithm}[H]
\caption{{\stronggreedy} for $(k,1)$-anonymization}
\begin{algorithmic}[1]
\STATE //Input: $k, G=(V,E)$
\STATE Compute $R = $ residual anonymity of $G$
\WHILE{$R > 0$}
\STATE Find edge $uv$ that maximally reduces $R$
\STATE $E = E \cup \{uv\}$
\STATE Update $R$
\ENDWHILE
\end{algorithmic}
\end{algorithm}

Since the added edges do not reinforce each other in the strong
$(k,1)$-anonymization problem, the analysis of {\tt Strong-Greedy} is
similar to the analysis of the greedy algorithm for the standard
set-cover problem.

\begin{theorem}
{\stronggreedy} is a polynomial-time $2\log n$-approximation
algorithm for the strong \an-anonymization problem.
\end{theorem}

\begin{proof}
Suppose the optimal solution adds $t$ edges, to reduce the
residual anonymity of the graph by at most $kn < n^2$. Since edges
of the solution do not reinforce each other, there must exist some
edge that reduces the residual anonymity of the graph by at least
a factor of $t$.

Therefore at each iteration of Algorithm~2, we greedily choose an
edge to add that must cause at least this much reduction in the
residual anonymity of the graph. The algorithm will thus terminate
after $r$ steps where ${(1-1/t)}^r \leq 1/(kn)$, or
$r=t\log(kn)\leq 2t\log n$.
\end{proof}

We next show that $\log n$ is the best factor we
could hope to achieve for unbounded $k$, for both the weak and strong
$(k,1)$-anonymization problems via an approximation-preserving reduction from the hitting set problem.

\begin{theorem}
The weak and strong $(k,1)$-anonymization problems with $k$ unbounded are
$\Omega(\log n)$-approximation NP-hard.
\end{theorem}

\begin{proof}
Hitting set is $\Omega(\log n)$-approximation NP-hard. Consider an
instance of the hitting-set problem consisting of sets ${\cal S} =
\{S_1, S_2, \ldots\}$. Let $k$ be greater than the maximum number
of sets intersecting any one set $S_i$. Add a unique element $v_i$
to each $S_i$. Additionally, construct sets ${\cal T} = \{T_1,
T_2, \ldots\}$ such that each $T_i$ contains the appropriate
$v_i$'s so that every $S_i$ intersects exactly $k-1$ other sets. In
every set $T_i$ add an additional element $w$ so that each set in
${\cal T}$ intersects at least $k$ other sets. Now construct a
bipartite graph $G=(U,V,E)$, where the vertices of $U$ correspond
to the sets in ${\cal S}$ and ${\cal T}$, the vertices of $V$
correspond to individual members of these sets, with $E$
indicating membership of elements from $V$ in sets from $U$. For
every element $u$ in $U$ create $(k+1)$ new vertices of degree $1$
in $V$ and connect them to $u$. In the resulting graph, the
vertices in $V$ are all \an-anonymous, however the vertices in $U$
that correspond to sets in ${\cal S}$ are only $(k-1,
1)$-anonymous. Consider the $t$ nodes in $V$ that are the optimal
solution to the hitting-set problem. Then matching these nodes
using $\lceil t/2\rceil$ edges will be an optimal solution to the
strong or weak $(k, 1)$-anonymization problem in the bipartite
graph $G=(U,V,E)$. Therefore, an optimal solution to the
anonymization problem corresponds to an optimal solution to the
hitting-set problem which is $\Omega(\log n)$-hard to approximate.
\end{proof}

\section{{\anonvar}-anonymization for $\ell > 1$}\label{K_L}

In this section we provide algorithms for the weak and strong
{\anonvar}-anonymization problems when $\ell > 1$. 

The algorithm for weak {\anonvar}-anonymization is a randomized algorithm that constructs a bounded-degree expander between deficient vertices. Given a
$(k,\ell')$-anonymous graph $G$, it solves the weak $(k,
\ell)$-anonymization problem by adding only ${\text
O}(\sqrt{k-k')\ell})$ additional edges at each vertex. The algorithm
can also be easily adapted to solve the weak $(k,
\ell)$-anonymization problem for any input graph irrespective of
its initial anonymity. 

\begin{theorem}
There exists a randomized polynomial-time algorithm that adds ${\text O}(\sqrt{(k-k')\ell})$ edges per vertex and increases the anonymity of
a graph from $(k',\ell)$ to $(k,\ell)$ where $\ell\leq k\leq
n^{1-\epsilon}$ and $\epsilon$ is a constant greater than $0$.
\end{theorem}

{\em Proof Sketch:} Randomly partition the $n$ vertices into $n/\ell$ sets
of size $\ell$. Treat each set as a ``supernode''.
Construct an expander of degree $\sqrt{(k-k')/\ell}$ on these
$n/\ell$ supernodes. In this way each supernode has $(k-k')\ell$
supernodes in its $2$-neighborhood that can be reached through just one
intermediate supernode.  Replace each edge $uv$ of this expander
with a $K_{\ell, \ell}$ clique of edges between the constituent
vertices of the supernodes $u$ and $v$. Thus each vertex now has $k-k'$
vertices in its $2$-neighborhood that can be reached through
an intermediate set of size $\ell$. Since $l\leq k\leq
n^{1-\epsilon}$, we can show that with high probability, none of
these $k-k'$ new vertices will coincide with the $k'$ vertices
previously in the node's $2$-neighborhood.\\

As a final result, we present the algorithm for strong {\anonvar}-anonymization. This algorithm is a generalization of the {\stronggreedy} algorithm (see Algorithm~2).
The difference is that instead of picking a single edge to add at
every iteration the algorithm picks edges in groups of size at
most $\ell$. At each iteration it picks the group that causes the largest
reduction in the residual anonymity of the graph. The pseudocode
is given in Algorithm~3.

\begin{algorithm}[H]
\caption{{\stronggreedy} for $(k, \ell)$-anonymization}
\begin{algorithmic}[1]
\STATE //Input: $k, \ell, G=(V,E)$ \STATE Compute $R = $ residual
anonymity of $G$ \WHILE{$R > 0$} \STATE Find set of edges ${\cal
E}$, with $|{\cal E}| \leq \ell$, that maximally reduces $R$
\STATE $E = E \cup \cal E$ \STATE Update $R$ \ENDWHILE
\end{algorithmic}
\end{algorithm}

We can state the following theorem for the approximation factor of
Algorithm~3 when $\ell$ is a constant.

\begin{theorem}
Consider $G=(V,E)$ to be the input graph to the strong {\anonvar}-anonymization problem.
Also assume $\ell$ is a constant. 
Let $t$ be the optimal number of edges that need to be added to solve the strong {\anonvar}-anonymization problem on $G$. Then Algorithm~3 is a
polynomial-time $O(t^{\ell-1}\log n)$-approximation algorithm.
\end{theorem}

{\em Proof Sketch:} In the $(k, \ell)$- anonymization problem,
groups of up to $\ell$ edges at a time incident at a single vertex
can reduce the residual anonymity of a vertex adjacent to the $\ell$
endpoints of these edges. The $t$ edges added by the optimal
solution define at most $t^{\ell}$ subsets of at most $\ell$ edges
incident to a single vertex. By selecting such subsets greedily as
in a set-cover problem we ultimately reduce the residual anonymity
of the graph to $0$ in ${\text O}(t^{\ell}\log n)$ steps. We can show that reinforcement effects between subsets of edges are taken care of. This proves the $O(t^{\ell}\log n)$ bound on the number of edges
selected. If $\ell$ is a small constant, the approximation factor may not be too large. Further, in practice this simple algorithm may perform better than this worst case bound indicates. 

\section{Conclusions}\label{sec:conclusions}
Motivated by recent studies on privacy-preserving graph releases,
we proposed a new definition of anonymity in graphs. We further
defined two new combinatorial problems arising from this
definition, studied their complexity and proposed simple,
efficient and intuitive algorithms for solving them.

The key idea behind our anonymization scheme was to enforce
that every node in the graph should share some number of its
neighbors with $k$ other nodes. The optimization problems we
defined ask for the minimum number of edges to be added to the
input graph so that the anonymization requirement is satisfied.
For these optimization problems we provided algorithms that solve
them exactly ($k=2$) or approximately ($k>2$).

An interesting avenue for future work would be to fully characterize the
kinds of attacks that our definition of anonymity protects
against, and to study the impact of our anonymization schemes on the utility of the graph release.

Finally, we believe that the combinatorial problems we have
studied in this paper are interesting in their own right, and may
also prove useful in other domains. For example, at a high level
there is a similarity between the problem we study in this paper
and the problem of constructing reliable graphs for, say, reliable
routing.

\bibliographystyle{plain}
\bibliography{graphanon}

\end{document}